\documentclass[runningheads]{llncs}
\usepackage[T1]{fontenc}
\usepackage{graphicx}
\usepackage{physics}
\usepackage{amsmath}
\usepackage{tikz}
\usepackage{mathdots}
\usepackage{yhmath}
\usepackage{cancel}
\usepackage{color}
\usepackage{siunitx}
\usepackage{array}
\usepackage{multirow}
\usepackage{amssymb}
\usepackage{gensymb}
\usepackage{tabularx}
\usepackage{extarrows}
\usepackage{booktabs}
\usetikzlibrary{fadings}
\usetikzlibrary{patterns}
\usetikzlibrary{shadows.blur}
\usetikzlibrary{shapes}
\newtheorem{observation}{Observation}
\newtheorem{openprob}{Open Problem}

\newcommand{\id}{\mathsf{id}}
\newcommand{\NA}{\mathsf{nact}}
\newcommand{\EA}{\mathsf{eact}}

\newcommand{\DFPC}{\textsc{Depth First Pointer Chasing}}
\newcommand{\PC}{\textsc{Pointer Chasing}}
\newcommand{\SYM}{\textsc{Symmetry}}
\newcommand{\CFF}{\textsc{C4-free\-ness}}

\newcommand{\UpB}{\big(\frac{n}{\log n}\big)^{1/4}}
\newcommand{\repeatthanks}{\textsuperscript{\thefootnote}}
\begin{document}

\title{Energy-Efficient Distributed Algorithms for Synchronous Networks\thanks{This work was performed during the visit of the first and last authors to Universidad de Chile, and to Universidad Adolfo Ibañez, Chile.}}
\titlerunning{Energy-Efficient Distributed Algorithms}
%
\author{
Pierre Fraigniaud\inst{1}\thanks{Additional support from ANR project DUCAT (ref. ANR-20-CE48-0006).} 
\and
Pedro Montealegre\inst{2}\thanks{Additional support from ANID  via PIA/Apoyo a Centros 
Cientificos y Tecnológicos de Excelencia AFB 170001, Fondecyt 1220142 and Fondecyt 1230599.}
\and
Ivan Rapaport\inst{3}\repeatthanks
\and \\
Ioan Todinca\inst{4}
}
\authorrunning{P. Fraigniaud, P. Montealegre, I. Rapaport,  I. Todinca}
%
\institute{
Institut de Recherche en Informatique Fondamentale (IRIF), CNRS and Université Paris Cité, Paris, France.
\email{pierre.fraigniaud@irif.fr}
\and
Facultad de Ingenier\'ia y Ciencias, Universidad Adolfo Ib\'a\~nez, Santiago, Chile
\email{p.montealegre@uai.cl}
\and
Departamento de Ingeniería Matemática - Centro de Modelamiento Matemático (UMI 2807 CNRS), Universidad de Chile, Santiago, Chile
\email{rapaport@dim.uchile.cl}
\and
Laboratoire d'informatique fondamentale d'Orléans (LIFO), Université d'Orléans, Orléans, France
\email{Ioan.Todinca@univ-orleans.fr}
}
\maketitle              
\begin{abstract}
We study the design of energy-efficient algorithms for the LOCAL and CONGEST models. Specifically, as a measure of complexity, we consider the maximum, taken over all the edges, or over all the nodes, of the number of rounds at which an edge, or a node, is active in the algorithm. We first show that every Turing-computable problem has a CONGEST algorithm with constant node-activation complexity, and therefore constant edge-activation complexity as well. That is, every node (resp., edge) is active in sending (resp., transmitting) messages for only $O(1)$ rounds  during the whole execution of the algorithm. In other words, every Turing-computable problem can be solved by an algorithm consuming the least possible energy. In the LOCAL model, the same holds obviously, but with the additional feature that the algorithm runs in $O(\mbox{poly}(n))$ rounds in $n$-node networks. However, we show that insisting on algorithms running in $O(\mbox{poly}(n))$ rounds in the CONGEST model comes with a severe cost in terms of energy. Namely, there are problems requiring $\Omega(\mbox{poly}(n))$ edge-activations (and thus $\Omega(\mbox{poly}(n))$ node-activations as well) in the CONGEST model  whenever solved by algorithms bounded to run in $O(\mbox{poly}(n))$ rounds. Finally, we demonstrate the existence of a sharp separation between the edge-activation complexity and the node-activation complexity in the CONGEST model, for algorithms bounded to run in $O(\mbox{poly}(n))$ rounds. Specifically,  under this constraint, there is a problem with $O(1)$ edge-activation complexity but $\tilde{\Omega}(n^{1/4})$ node-activation complexity. 

\end{abstract}

\section{Introduction}

\subsection{Objective}

Designing computing environments consuming a limited amount of energy while achieving computationally complex tasks is an objective of utmost importance, especially in distributed systems involving a large number of computing entities. In this paper, we aim at designing energy-efficient algorithms for the standard LOCAL and CONGEST models of distributed computing in networks~\cite{peleg2000}. Both models assume a network modeled as an $n$-node graph $G=(V,E)$, where each node is provided with an identifier, i.e., an integer that is unique in the network, which can be stored on $O(\log n)$ bits. All nodes are assumed to run the same algorithm, and computation proceeds as a series of synchronous rounds (all nodes start simultaneously at round~1). During a round, every node sends a message to each of its neighbors, receives the messages sent by its neighbors, and performs some individual computation. The two models LOCAL and CONGEST differ only in the amount of information that can be exchanged between nodes at each round. 

The LOCAL model does not bound the size of the messages, whereas the CONGEST model allows only messages of size $O(\log n)$ bits. Initially, every node~$v\in V$ knows solely its identifier~$\id(v)$, an upper bound of the number~$n$ of nodes, which is assumed to be polynomial in~$n$ and to be the same for all nodes, plus possibly some input bit-string $x(v)$ depending on the task to be solved by the nodes. In this paper, we denote by $N$ the maximum between the largest identifier and the upper bound on~$n$ given to all nodes. Hence $N=O(\mbox{poly}(n))$, and is supposed to be known by all nodes. After a certain number of rounds, every node outputs a bit-string~$y(v)$, where the correctness of the collection of outputs $y=\{y(v):v\in V\}$ is defined with respect to the specification of the task to be solved, and may depend on the collection of inputs $x=\{x(v):v\in V\}$ given to the nodes, as well as on the graph~$G$ (but not on the identifiers assigned to the nodes, nor on the upper bound~$N$). 

\paragraph{Activation complexity.} 

We measure the energy consumption of an algorithm~$A$ by counting how many times each node and each edge is activated during the execution of the algorithm. More specifically, a node~$v$ (resp., an edge~$e$) is said to be \emph{active} at a given round~$r$ if $v$ is sending a message to at least one of its neighbors at round~$r$ (resp., if a message traverses $e$ at round~$r$). The \emph{node-activation} and the \emph{edge-activation} of an algorithm~$A$ running in a graph $G=(V,E)$ are respectively  defined as 
\[
\NA(A):=\max_{v\in V}\#\mbox{activation}(v),
\;\; \mbox{and} \;\;
\EA(A):=\max_{e\in E}\#\mbox{activation}(e),
\]
where $\#\mbox{activation}(v)$ (resp., $\#\mbox{activation}(e)$) denotes the number of rounds during which node~$v$ (resp., edge~$e$) is active along the execution of the algorithm~$A$.
By definition, we have that, in any graph of maximum degree~$\Delta$, 
\begin{equation}\label{eq:ineq-activation}
\EA(A)\leq 2\cdot \NA(A), ~\textrm{ and }~ \NA(A) \leq \Delta\cdot\EA(A).
\end{equation}

\paragraph{Objective.} 

Our goal is to design \emph{frugal} algorithms, that is, algorithms with \emph{constant} node-activation, or to the least \emph{constant} edge-activation, independent of the number~$n$ of nodes and of the number~$m$ of edges. Indeed, such algorithms can be viewed as consuming the least possible energy for solving a given task. Moreover, even if the energy requirement for solving the task naturally grows with the number of components (nodes or edges) of the network, it grows \emph{linearly} with this number whenever using frugal algorithms. We refer to \emph{node-frugality} or \emph{edge-frugality} depending on whether we focus on node-activation or edge-activation, respectively. 

\subsection{Our Results}

We first show that every Turing-computable problem\footnote{A problem is Turing-computable if there exists a Turing machine that, given any graph with identifiers and inputs assigned to the nodes, computes the output of each node in the graph.} can thus be solved by a node-frugal algorithm in the LOCAL model as well as in the CONGEST model. It follows from Eq.~\ref{eq:ineq-activation} that every Turing-computable problem can be solved by an edge-frugal algorithm in both models. In other words, every problem can be solved by an energy-efficient distributed algorithm. One important question remains: what is the round complexity of frugal algorithms? 

In the LOCAL model, our node-frugal algorithms run in $O(\mbox{poly}(n))$ rounds. However, they may run in exponentially many rounds in the CONGEST model. We show that this cannot be avoided. Indeed, even if many symmetry-breaking problems such as computing a maximal-independent set ({\sc mis}) and computing a $(\Delta+1)$-coloring can be solved by a node-frugal algorithm performing in $O(\mbox{poly}(n))$ rounds, we show that there exist problems (e.g., deciding $C_4$-freeness or deciding the presence of symmetries in the graph) that cannot be solved in $O(\mbox{poly}(n))$ rounds in the CONGEST model by any edge-frugal algorithm. 

Finally, we discuss the relation between node-activation complexity and edge-activation complexity. We show that the bounds given by Eq.~\ref{eq:ineq-activation} are essentially the best that can be achieved in general. Precisely, we identify a problem, namely \DFPC{} (\textsc{dfpc}), which has edge-activation complexity $O(1)$ for all graphs with an algorithm running in $O(\mbox{poly}(n))$ rounds in the CONGEST model, but satisfying that, for every $\Delta =O\left(\UpB\right)$, its node-activation complexity in graphs with maximum degree $\Delta$ is $\Omega(\Delta)$ whenever solved by an algorithm bounded to run in $O(\mbox{poly}(n))$ rounds in the CONGEST model. In particular, \DFPC{} has constant edge-activation complexity but node-activation complexity $\tilde{\Omega}(n^{1/4})$ in the CONGEST model (for $O(\mbox{poly}(n))$-round algorithms). 

Our main results are summarized in Table~\ref{tab:summary}. 

\begin{table}
\begin{center} 
\begin{tabular}{l|l|l|l|} 
& \hspace{1cm} \textbf{Awakeness}
& \hspace{.1cm} \textbf{Node-Activation}
& \hspace{.3cm} \textbf{Edge-Activation} \\
\hline
LOCAL 
& $\bullet \;\forall \Pi, \Pi \in O(\log n)$ with  
& $\bullet \;\forall \Pi, \Pi \in O(1)$ with 
& $\bullet \; \forall \Pi, \Pi \in O(1)$ with\\
& $\;\; O(\mbox{poly}(n))$ rounds \cite{BarenboimM21} 
& $\;\; O(\mbox{poly}(n))$ rounds 
& $\;\; O(\mbox{poly}(n))$ rounds\\
& $\bullet \; \mbox{\sc st}\in \Omega(\log n)$ \cite{BarenboimM21} 
& & \\
\hline 
CONGEST 
& $\bullet \; \mbox{\sc mis}\in O(\mbox{polyloglog}(n))$  
& $\bullet \; \forall \Pi, \Pi\in O(1)$ 
& $\bullet \; \forall \Pi, \Pi \in O(1)$\\
& \;\; with $O(\mbox{polylog}(n))$  
& $\bullet \;\mbox{poly}(n)$ rounds 
& $\bullet \; \mbox{poly}(n) $ rounds \\
& \;\; rounds \cite{DufoulonMP22} (randomized)
& $\;\; \Rightarrow \exists \Pi \in \Omega(\mbox{poly}(n)) $ 
& $\;\; \Rightarrow \exists \Pi \in \Omega(\mbox{poly}(n)) $ \\
& $\bullet \; \mbox{\sc mst} \in O(\log n)$
& $\bullet \;\mbox{poly}(n)$ rounds 
& $\bullet \;\mbox{\sc dfpc}\in O(1)$ with \\
& \;\; with $O(\mbox{poly}(n))$
& $\;\; \Rightarrow \mbox{\sc dfpc}\in \tilde{\Omega}(n^{1/4})$
& $\; \; O(\mbox{poly}(n))$ rounds \\
& \;\; rounds \cite{AugustineMP22}
&
& $\bullet \;\Pi\in \mbox{FO and}\; \Delta=O(1) $\\
&
&
& $\;\; \Rightarrow \Pi\in O(1) \; \mbox{with}$ \\
&
&
& $\;\; O(\mbox{poly}(n))$ rounds~\cite{GrumbachW09} \\
\hline
\end{tabular}
\end{center}
\caption{\sl Summary of our results where, for a problem $\Pi$, $\Pi\in O(f(n))$ means that the corresponding complexity of $\Pi$ is $O(f(n))$ (same shortcut for $\Omega$).}
\label{tab:summary}
\end{table}

\paragraph{Our Techniques.} 

Our upper bounds are mostly based on similar types of upper bounds techniques used in the sleeping model~\cite{BarenboimM21,ChatterjeeGP20} (cf. Section~\ref{subsec:related-work}), based on constructing spanning trees along with
gathered and broadcasted information. However, the models considered in this paper do not suffer from the same limitations as the sleeping model (cf. Section~\ref{sec:preliminaries}), and thus one can achieve activation complexity $O(1)$ in scenarios where the sleeping model limits the awake complexity to $\Omega(\log n)$. 

Our lower bounds for CONGEST are based on reductions from 2-party communication complexity. However, as opposed to the standard CONGEST model in which the simulation of a distributed algorithm by two players is straightforward (each player performs the rounds sequentially, one by one, and exchanges the messages sent across the cut between the two subsets of nodes handled by the players at each round), the simulation of distributed algorithms in which only subsets of nodes are active at various rounds requires more care. This is especially the case when the simulation must not only control the amount of information exchanged between these players, but also the number of communication steps performed by the two players. Indeed, there are 2-party communication complexity problems that are hard for $k$ steps, but trivial for $k+1$ steps~\cite{NisanW93}, and some of our lower bounds rely on this fact.

\subsection{Related Work}
\label{subsec:related-work}

The study of frugal algorithms has been initiated in~\cite{GrumbachW09}, which focuses on the edge-frugality in the CONGEST model. It is shown that for {\emph{bounded-degree graphs}}, any problem expressible in first-order logic (e.g., $C_4$-freeness) can be solved by an edge-frugal algorithm running in $O(\mbox{poly}(n))$ rounds in the CONGEST model. This also holds for planar graphs with no bounds on the maximum degree, whenever the nodes are provided with their local combinatorial embedding. Our results show that these statements cannot be extended to arbitrary graphs as we prove that any algorithm solving $C_4$-freeness in $O(\mbox{poly}(n))$ rounds in the CONGEST model has edge-activation $\tilde{\Omega}(\sqrt{n})$. 

More generally, the study of energy-efficient algorithms in the context of distributed computing in networks has been previously considered in the framework of the \emph{sleeping} model, introduced in~\cite{ChatterjeeGP20}. This model assumes that nodes can be in two states: \emph{awake} and \emph{asleep}. A node in the awake state performs as in the LOCAL and CONGEST models, but may also decide to fall asleep, for a prescribed amount of rounds, controlled by each node, and depending on the algorithm executed at the nodes. A sleeping node is totally inactive in the sense that it does not send messages, it cannot receive messages (i.e., if a message is sent to a sleeping node by an awake neighbor, then the message is lost), and it is computationally idle  (apart from counting rounds). The main measure of interest in the sleeping model is the \emph{awake complexity}, defined as the maximum, taken over all nodes, of the number of rounds at which each node is awake during the execution of the algorithm. 

In the LOCAL model, it is known~\cite{BarenboimM21} that all problems have awake complexity $O(\log n)$, using algorithms running in $O(\mbox{poly}(n))$ rounds. This bound is tight in the sense that there are problems (e.g., spanning tree construction) with awake complexity $\Omega(\log n)$ \cite{BarenboimM21,ChangDHHLP18}. 

In the CONGEST model, It was first shown~\cite{ChatterjeeGP20} that {\sc mis} has constant \emph{average} awake complexity, thanks to a \emph{randomized} algorithm running in $O(\mbox{polylog}(n))$ rounds. The round complexity was improved in~\cite{GhaffariP22} with a \emph{randomized} algorithm running in $O(\log n)$ rounds. The (worst-case) awake complexity of {\sc mis} was proved to be $O(\log\log n)$ using a  \emph{randomized} Monte-Carlo algorithm running in $O(\mbox{poly}(n))$ rounds~\cite{DufoulonMP22}. This (randomized) round complexity can even be reduced to $O(\log^3n \cdot \log\log n \cdot \log^\star n)$, at the cost of slightly increasing the awake complexity to $O(\log\log n \cdot \log^\star n)$. {\sc mst} has also been considered, and it was proved~\cite{AugustineMP22} that its (worst-case) awake complexity is $O(\log n)$ thanks to a (deterministic) algorithm running in $O(\mbox{poly}(n))$ rounds. The upper bound on the awake complexity of {\sc mst} is tight, thank to the lower bound for spanning tree ({\sc st}) in~\cite{BarenboimM21}. 

\section{Preliminaries}
\label{sec:preliminaries}

In this section, we illustrate the difference between the standard LOCAL and CONGEST models, their sleeping variants, and our node- and edge-activation variants. Fig.~\ref{fig:1}(a) displays the automaton corresponding to the behavior of a node in the standard models. A node is either \emph{active}~(A) or \emph{terminated}~(T). At each clock tick (i.e., round) a node is subject to message events corresponding to sending and receiving messages to/from neighbors. A node remains active until it terminates. 

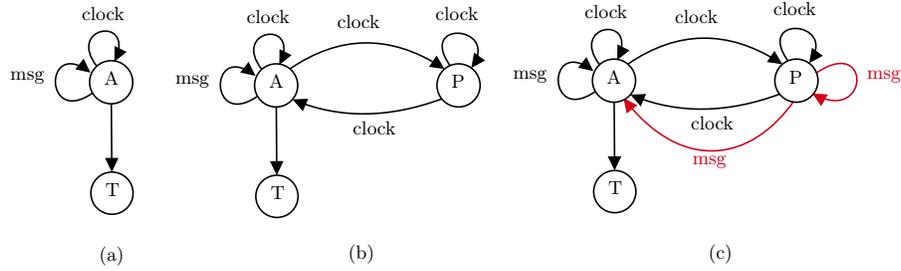
\begin{figure}[!h]
    \centering
    \scalebox{0.8}{
    \tikzset{every picture/.style={line width=0.75pt}} 

\begin{tikzpicture}[x=0.75pt,y=0.75pt,yscale=-1,xscale=1]

\draw    (384,180) .. controls (364.62,155.02) and (406.85,148.15) .. (396.15,174.46) ;
\draw [shift={(395,177)}, rotate = 296.39] [fill={rgb, 255:red, 0; green, 0; blue, 0 }  ][line width=0.08]  [draw opacity=0] (8.93,-4.29) -- (0,0) -- (8.93,4.29) -- cycle    ;
\draw    (503,178) .. controls (481.68,148.65) and (532.32,152.01) .. (517.56,176.65) ;
\draw [shift={(516,179)}, rotate = 306.11] [fill={rgb, 255:red, 0; green, 0; blue, 0 }  ][line width=0.08]  [draw opacity=0] (8.93,-4.29) -- (0,0) -- (8.93,4.29) -- cycle    ;
\draw    (381,197) .. controls (350.79,217.71) and (350.51,163.57) .. (377.84,182.4) ;
\draw [shift={(380,184)}, rotate = 218.26] [fill={rgb, 255:red, 0; green, 0; blue, 0 }  ][line width=0.08]  [draw opacity=0] (8.93,-4.29) -- (0,0) -- (8.93,4.29) -- cycle    ;
\draw [color={rgb, 255:red, 208; green, 2; blue, 27 }  ,draw opacity=1 ]   (520.71,199.44) .. controls (555.04,216.25) and (555.07,158.96) .. (518,187) ;
\draw [shift={(518,198)}, rotate = 29.85] [fill={rgb, 255:red, 208; green, 2; blue, 27 }  ,fill opacity=1 ][line width=0.08]  [draw opacity=0] (8.93,-4.29) -- (0,0) -- (8.93,4.29) -- cycle    ;
\draw [color={rgb, 255:red, 208; green, 2; blue, 27 }  ,draw opacity=1 ]   (401.11,203.71) .. controls (435.79,247.07) and (479.09,242.42) .. (507,202) ;
\draw [shift={(399,201)}, rotate = 52.87] [fill={rgb, 255:red, 208; green, 2; blue, 27 }  ,fill opacity=1 ][line width=0.08]  [draw opacity=0] (8.93,-4.29) -- (0,0) -- (8.93,4.29) -- cycle    ;
\draw    (171,183) .. controls (151.62,158.02) and (193.85,151.15) .. (183.15,177.46) ;
\draw [shift={(182,180)}, rotate = 296.39] [fill={rgb, 255:red, 0; green, 0; blue, 0 }  ][line width=0.08]  [draw opacity=0] (8.93,-4.29) -- (0,0) -- (8.93,4.29) -- cycle    ;
\draw    (290,181) .. controls (268.68,151.65) and (319.32,155.01) .. (304.56,179.65) ;
\draw [shift={(303,182)}, rotate = 306.11] [fill={rgb, 255:red, 0; green, 0; blue, 0 }  ][line width=0.08]  [draw opacity=0] (8.93,-4.29) -- (0,0) -- (8.93,4.29) -- cycle    ;
\draw    (168,200) .. controls (137.79,220.71) and (137.51,166.57) .. (164.84,185.4) ;
\draw [shift={(167,187)}, rotate = 218.26] [fill={rgb, 255:red, 0; green, 0; blue, 0 }  ][line width=0.08]  [draw opacity=0] (8.93,-4.29) -- (0,0) -- (8.93,4.29) -- cycle    ;
\draw    (67,181) .. controls (47.62,156.02) and (89.85,149.15) .. (79.15,175.46) ;
\draw [shift={(78,178)}, rotate = 296.39] [fill={rgb, 255:red, 0; green, 0; blue, 0 }  ][line width=0.08]  [draw opacity=0] (8.93,-4.29) -- (0,0) -- (8.93,4.29) -- cycle    ;
\draw    (64,198) .. controls (33.79,218.71) and (33.51,164.57) .. (60.84,183.4) ;
\draw [shift={(63,185)}, rotate = 218.26] [fill={rgb, 255:red, 0; green, 0; blue, 0 }  ][line width=0.08]  [draw opacity=0] (8.93,-4.29) -- (0,0) -- (8.93,4.29) -- cycle    ;

\draw  [fill={rgb, 255:red, 255; green, 255; blue, 255 }  ,fill opacity=1 ]  (393, 190) circle [x radius= 13.6, y radius= 13.6]   ;
\draw (387,182) node [anchor=north west][inner sep=0.75pt]   [align=left] {A};
\draw  [fill={rgb, 255:red, 255; green, 255; blue, 255 }  ,fill opacity=1 ]  (508, 190) circle [x radius= 13.6, y radius= 13.6]   ;
\draw (502,182) node [anchor=north west][inner sep=0.75pt]   [align=left] {P};
\draw  [fill={rgb, 255:red, 255; green, 255; blue, 255 }  ,fill opacity=1 ]  (393.5, 260) circle [x radius= 13.31, y radius= 13.31]   ;
\draw (388,252) node [anchor=north west][inner sep=0.75pt]   [align=left] {T};
\draw (373,141) node [anchor=north west][inner sep=0.75pt]  [font=\small] [align=left] {clock};
\draw (441,236) node [anchor=north west][inner sep=0.75pt]  [font=\small,color={rgb, 255:red, 208; green, 2; blue, 27 }  ,opacity=1 ] [align=left] {msg};
\draw (551,182) node [anchor=north west][inner sep=0.75pt]  [font=\small,color={rgb, 255:red, 208; green, 2; blue, 27 }  ,opacity=1 ] [align=left] {msg};
\draw (328,182) node [anchor=north west][inner sep=0.75pt]  [font=\small] [align=left] {msg};
\draw (440,211) node [anchor=north west][inner sep=0.75pt]  [font=\small] [align=left] {clock};
\draw (492,139) node [anchor=north west][inner sep=0.75pt]  [font=\small] [align=left] {clock};
\draw (430,144) node [anchor=north west][inner sep=0.75pt]  [font=\small] [align=left] {clock};
\draw  [fill={rgb, 255:red, 255; green, 255; blue, 255 }  ,fill opacity=1 ]  (180, 193) circle [x radius= 13.6, y radius= 13.6]   ;
\draw (174,185) node [anchor=north west][inner sep=0.75pt]   [align=left] {A};
\draw  [fill={rgb, 255:red, 255; green, 255; blue, 255 }  ,fill opacity=1 ]  (295, 193) circle [x radius= 13.6, y radius= 13.6]   ;
\draw (289,185) node [anchor=north west][inner sep=0.75pt]   [align=left] {P};
\draw  [fill={rgb, 255:red, 255; green, 255; blue, 255 }  ,fill opacity=1 ]  (180.5, 263) circle [x radius= 13.31, y radius= 13.31]   ;
\draw (175,255) node [anchor=north west][inner sep=0.75pt]   [align=left] {T};
\draw (160,144) node [anchor=north west][inner sep=0.75pt]  [font=\small] [align=left] {clock};
\draw (115,185) node [anchor=north west][inner sep=0.75pt]  [font=\small] [align=left] {msg};
\draw (227,214) node [anchor=north west][inner sep=0.75pt]  [font=\small] [align=left] {clock};
\draw (279,142) node [anchor=north west][inner sep=0.75pt]  [font=\small] [align=left] {clock};
\draw (217,147) node [anchor=north west][inner sep=0.75pt]  [font=\small] [align=left] {clock};
\draw  [fill={rgb, 255:red, 255; green, 255; blue, 255 }  ,fill opacity=1 ]  (76, 191) circle [x radius= 13.6, y radius= 13.6]   ;
\draw (70,183) node [anchor=north west][inner sep=0.75pt]   [align=left] {A};
\draw  [fill={rgb, 255:red, 255; green, 255; blue, 255 }  ,fill opacity=1 ]  (76.5, 261) circle [x radius= 13.31, y radius= 13.31]   ;
\draw (71,253) node [anchor=north west][inner sep=0.75pt]   [align=left] {T};
\draw (56,142) node [anchor=north west][inner sep=0.75pt]  [font=\small] [align=left] {clock};
\draw (11,183) node [anchor=north west][inner sep=0.75pt]  [font=\small] [align=left] {msg};
\draw (67,293) node [anchor=north west][inner sep=0.75pt]   [align=left] {(a)};
\draw (224,292) node [anchor=north west][inner sep=0.75pt]   [align=left] {(b)};
\draw (451,292) node [anchor=north west][inner sep=0.75pt]   [align=left] {(c)};
\draw    (401.73,179.57) .. controls (432.88,158.24) and (464.54,157.73) .. (496.7,178.04) ;
\draw [shift={(499.18,179.65)}, rotate = 213.59] [fill={rgb, 255:red, 0; green, 0; blue, 0 }  ][line width=0.08]  [draw opacity=0] (8.93,-4.29) -- (0,0) -- (8.93,4.29) -- cycle    ;
\draw [color={rgb, 255:red, 0; green, 0; blue, 0 }  ,draw opacity=1 ]   (497.4,198.52) .. controls (464.5,210.63) and (434,211.09) .. (405.92,199.92) ;
\draw [shift={(403.33,198.85)}, rotate = 23.06] [fill={rgb, 255:red, 0; green, 0; blue, 0 }  ,fill opacity=1 ][line width=0.08]  [draw opacity=0] (8.93,-4.29) -- (0,0) -- (8.93,4.29) -- cycle    ;
\draw    (393.1,203.6) -- (393.38,243.69) ;
\draw [shift={(393.4,246.69)}, rotate = 269.59] [fill={rgb, 255:red, 0; green, 0; blue, 0 }  ][line width=0.08]  [draw opacity=0] (8.93,-4.29) -- (0,0) -- (8.93,4.29) -- cycle    ;
\draw    (188.73,182.57) .. controls (219.88,161.24) and (251.54,160.73) .. (283.7,181.04) ;
\draw [shift={(286.18,182.65)}, rotate = 213.59] [fill={rgb, 255:red, 0; green, 0; blue, 0 }  ][line width=0.08]  [draw opacity=0] (8.93,-4.29) -- (0,0) -- (8.93,4.29) -- cycle    ;
\draw [color={rgb, 255:red, 0; green, 0; blue, 0 }  ,draw opacity=1 ]   (284.4,201.52) .. controls (251.5,213.63) and (221,214.09) .. (192.92,202.92) ;
\draw [shift={(190.33,201.85)}, rotate = 23.06] [fill={rgb, 255:red, 0; green, 0; blue, 0 }  ,fill opacity=1 ][line width=0.08]  [draw opacity=0] (8.93,-4.29) -- (0,0) -- (8.93,4.29) -- cycle    ;
\draw    (180.1,206.6) -- (180.38,246.69) ;
\draw [shift={(180.4,249.69)}, rotate = 269.59] [fill={rgb, 255:red, 0; green, 0; blue, 0 }  ][line width=0.08]  [draw opacity=0] (8.93,-4.29) -- (0,0) -- (8.93,4.29) -- cycle    ;
\draw    (76.1,204.6) -- (76.38,244.69) ;
\draw [shift={(76.4,247.69)}, rotate = 269.59] [fill={rgb, 255:red, 0; green, 0; blue, 0 }  ][line width=0.08]  [draw opacity=0] (8.93,-4.29) -- (0,0) -- (8.93,4.29) -- cycle    ;

\end{tikzpicture}}
    \caption{(a) Classical model (b) Sleeping model, (c) Activation model.  }
    \label{fig:1}
\end{figure}

Fig.~\ref{fig:1}(b) displays the automaton corresponding to the behavior of a node in the sleeping variant. In this variant, a node can also be in a \emph{passive} (P) state. In this state, the clock event can either leave the node passive, or awake the node, which then moves back to the active state.  

Finally, Fig.~\ref{fig:1}(c) displays the automaton corresponding to the behavior of a node in our  activation variants. It differs from the sleeping variant in that a passive node is also subject to message events, which can leave the node passive, but may also move the node to the active state.  In particular, a node does not need to be active for receiving messages, and incoming messages may not trigger an immediate response from the node (e.g., forwarding information). Instead, a node can remain passive while collecting information from each of its neighbors, and eventually react by becoming active. 

\paragraph{Example 1: Broadcast.} 

Assume that one node of the $n$-node cycle~$C_n$ has a token to be broadcast to all the nodes. Initially, all nodes are active. However, all nodes but the one with the token become immediately passive when the clock ticks for entering the second round. The node with the token sends the token to one of its neighbors, and becomes passive at the next clock tick. Upon reception of the token, a passive node becomes active, forwards the token, and terminates. When the source node receives the token back, it becomes active, and terminates. The node-activation complexity of broadcast is therefore~$O(1)$, whereas it is known that broadcasting has awake complexity $\Omega(\log n)$ in the sleeping model~\cite{BarenboimM21}. 

\paragraph{Example 2: At-least-one-leader.} 

Assume that each node of the cycle~$C_n$ has an input-bit specifying whether the node is leader or not, and the nodes must collectively check that there is at least one leader. Every leader broadcasts a token, outputs accept, and terminates. Non-leader nodes become passive immediately after the beginning of the algorithm, and start waiting for $N$~rounds (recall that $N$ is an upper bound on the number~$n$ of nodes). Whenever the ``sleep'' of a (passive) non-leader is interrupted by the reception of a token, it becomes active, forwards the token, outputs accept, and terminates. After $N$ rounds, a passive node that has not been ``awaken'' by a token becomes active, outputs reject, and terminates. This guarantees that there is at least one leader if and only if all nodes accept. The node-activation complexity of this algorithm is $O(1)$, while the awake complexity of at-least-one-leader is $\Omega(\log n)$ in the sleeping model, by reduction to broadcast. 

\medskip 

The following observation holds for LOCAL and CONGEST, by noticing that every algorithm for the sleeping model can be implemented with no overheads in terms of node-activation.

\begin{observation}\label{obs:sleep-vs-activation}
In $n$-node graphs, every algorithm with awake complexity~$a(n)$ and round complexity~$r(n)$ has node-activation complexity at most~$a(n)$ and round complexity at most~$r(n)$. 
\end{observation}

It follows from Observation~\ref{obs:sleep-vs-activation} that all upper bound results for the awake complexity directly transfer to the node-activation complexity. However, as we shall show in this paper, in contrast to the sleeping model in which some problems (e.g., spanning tree) have awake complexity $\Omega(\log n)$, even in the LOCAL model, all problems admit a frugal algorithm in the CONGEST model, i.e., an algorithm with node-activation~$O(1)$. 

\begin{definition}
A LOCAL or CONGEST algorithm is \emph{node-frugal} (resp., \emph{edge-frugal}) if the activation of every node (resp., edge) is upper-bounded by a constant independent of the graph, and of the identifiers and inputs given to the nodes. 
\end{definition}

\section{Universality of Frugal Algorithms}

In this section we show that every Turing-computable problem can be solved by frugal algorithms, both in the LOCAL and  CONGEST models. Thanks to Eq.~\ref{eq:ineq-activation}, it is sufficient to prove that this holds for node-frugality.

\begin{lemma} \label{lemm:leader}
There exists a CONGEST algorithm electing a leader, and constructing a BFS tree rooted at the leader, with node-activation complexity~$O(1)$, and performing in $O(N^2) = O(\mbox{\rm poly}(n))$ rounds.
\end{lemma}

\begin{proof}
The algorithm elects as leader the node with smallest identifier, and initiates a breadth-first search from that node. At every node~$v$, the protocol performs as follows. 

\begin{itemize}
\item If $v$ has received no messages until round $\id(v) \cdot N$, then $v$ elects itself as leader, and starts a BFS by sending message $(\id(v),0)$ to all its neighbors. Locally, $v$~sets its parent in the BFS tree to \(\bot\), and the distance to the root to~$0$.

\item Otherwise, let $r$ be the first round at which vertex~$v$  receives a message. Such a message is of type $(\id(u),d)$ where $u$ is the neighbor of $v$ which sent the message to~$v$, and $d$ is the distance from $u$ to the leader in the graph. Node $v$ sets its parent in the BFS tree to $\id(u)$, its distance to the root to $d+1$, and, at round $r+1$, it sends the message $(\id(v),d+1)$ to all its neighbors. (If $v$ receives several messages at round~$r$, from different neighbors, then $v$ selects the messages coming from the neighobors with smallest identifier).
\end{itemize}

The node $v$ with smallest identifier is indeed the node initiating the BFS, as the whole BFS is constructed between
rounds ${\id(v)\cdot N}$ and ${\id(v)\cdot N + N - 1}$, and $N\geq n$. The algorithm terminates at round at most~$O(N^2)$. 
\qed
\end{proof}

An instance of a problem is a triple $(G,\id,x)$ where $G=(V,E)$ is an $n$-node graph, $\id:V\to [1,N]$ with $N=O(\mbox{poly}(n))$, and $x:V\to [1,\nu]$ is the input assignment to the nodes. Note that the input range $\nu$ may depend on~$n$, and even be exponential in~$n$, even for classical problems, e.g., whenever weights assigned to the edges are part of the input. A solution to a graph problem is an output assignment $y:V\to [1,\mu]$, and the correctness of~$y$ depends on $G$ and $x$ only, with respect to the specification of the problem. We assume that $\mu$ and $\nu$ are initially known to the nodes, as it is the case for, e.g., {\sc mst}, in which the weights of the edges can be encoded on $O(\log n)$ bits.

\begin{theorem}\label{theo:all-in-LOCAL}
Every Turing-computable problem
has a LOCAL algorithm with $O(1)$ node-activation complexity, and running in  $O(N^2) = O(\mbox{\rm poly}(n))$ rounds.
\end{theorem}

\begin{proof}
Once the BFS tree~$T$ of Lemma \ref{lemm:leader} is constructed, the root can (1)~gather the whole instance $(G,\id,x)$, (2)~compute a solution~$y$, and (3)~broadcast~$y$ to all nodes. Specifically, every leaf~$v$ of~$T$ sends the set 
\[
E(v)=\big\{\{(\id(v),x(v)),(\id(w),x(w))\}:w\in N(v)\big\}
\]
to its parent in~$T$. An internal node~$v$ waits for receiving a set of edges $S(u)$ from each of its children~$u$ in~$T$, and then forwards the set 
\[
S(v)=E(v)\cup(\cup_{u\in \mbox{\tiny child}(v)}S(u))
\]
to its parent. This set can be encoded in \(O(N^2)\) bits by the adjacency matrix of the subgraph induced by the edges in \(S(v)\). Each node of $T$ is activated once during this phase, and thus the node-activation complexity of gathering is~1. Broadcasting the solution~$y$ from the leader to all the nodes is achieved along the edges of~$T$, again with node-activation~1.
\qed
\end{proof}

The algorithm used in the proof of Theorem~\ref{theo:all-in-LOCAL} cannot be implemented in CONGEST due to the size of the messages, which may require each node to be activated more than a constant number of times. To keep the node-activation constant, we increased the round complexity of the algorithm. 

\begin{lemma}\label{lem:simulation-local-congest}
Every node-frugal algorithm $\mathcal{A}$ performing in $R$ rounds in the LOCAL model with messages of size at most~$M$ bits\footnote{Without loss of generality, in \(\mathcal{A}\) each node sends the same message to all its neighbors at each round when it is active. Otherwise, the different messages can be merged into one, by adding the identifiers of the neighbors.} can be implemented by a node-frugal algorithm $\mathcal{B}$ performing in $R\,2^M$ rounds in the CONGEST model.
\end{lemma} 

\begin{proof}
Let $v$ be a node sending a message~$m$ through an incident edge~$e$ at round~$r$ of~$\mathcal{A}$. Then, in~$\mathcal{B}$, $v$ sends one ``beep'' through edge~$e$ at round $r\,2^M+t$ where $t$ is lexicographic rank of $m$ among the at most $2^M$ messages generated by~$\mathcal{A}$. 
\qed
\end{proof}

\begin{theorem}\label{theo:all-in-CONGEST}
Every Turing-computable problem
has a CONGEST algorithm with $O(1)$ node-activation complexity, and 
 running in 
$2^{\mbox{\rm\scriptsize poly}(n)(1+\log(\nu\mu))}$ rounds for inputs in the range $[1,\nu]$ and outputs in the range $[1,\mu]$. 
\end{theorem}

\begin{proof}
The algorithm used in the proof of Theorem~\ref{theo:all-in-LOCAL} used messages of size at most $N^2+N\log \nu$ bits during the gathering phase, and of size at most $N\log \mu$ bits during the broadcast phase. The result follows from Lemma~\ref{lem:simulation-local-congest}.\qed
\end{proof}

Of course, there are many problems that can be solved in the CONGEST model by a frugal algorithm much faster than the bound from Theorem~\ref{theo:all-in-CONGEST}. This is typically the case of all problems that can be solved by a sequential greedy algorithm visiting the nodes in arbitrary order, and producing a solution at the currently visited node based only on the partial solution in the neighborhood of the node. Examples of such problems are {\sc mis}, $\Delta + 1$-coloring, etc. We call such problem \emph{sequential-greedy}. 

\begin{theorem}
Every sequential-greedy problem whose solution at every node can be encoded on $O(\log n)$ bits has a node-frugal CONGEST algorithm running in $O(N) = O({\mbox{\rm poly}(n)})$ rounds.
\end{theorem}

\begin{proof}
Every node $v \in V$ generates its output  at round~$\id(v)$ according to its current knowledge about its neighborhood, and sends this output to all its neighbors. 
\qed
\end{proof}

\section{Limits of CONGEST Algorithms with Polynomially Many Rounds}

Given a graph \(G = (V,E)\) such that \(V\) is partitioned in two sets \(V_A, V_B\), the set of edges with one endpoint in \(V_A\) and the other in \(V_B\) is called the \emph{cut}. We denote by \(e(V_A,V_B)\) the number of edges in the cut, and by \(n(V_A,V_B)\) the number of nodes incident to an edge of the cut.  Consider the situation where there are two players, namely Alice and Bob. We say that a player controls a node \(v\) if it knows all its incident edges and its input. For a CONGEST algorithm \(\mathcal{A}\), we denote \(\mathcal{A}(\mathcal{I})\) the output of \(\mathcal{A}\) on input \(\mathcal{I} = (G,\id, x)\). We denote \(R_\mathcal{A}(n)\) the round complexity of \(\mathcal{A}\) on inputs of size \(n\).

\begin{lemma}[Simulation lemma]\label{lem:simu}
    Let  \(\mathcal{A}\) be an algorithm in the CONGEST model, let \(\mathcal{I} = (G,\id,x)\) be an input for $\mathcal{A}$, and let \(V_A, V_B\) be a partition of \(V(G)\).  Suppose that Alice controls all the nodes in \(V_A\), and Bob controls all the nodes in \(V_B\).  Then, there exists a communication protocol \(\mathcal{P}\) between Alice and Bob with at most \(2\cdot\min(n(V_A,V_B)\cdot \NA(\mathcal{A}), e(V_A,V_B)\cdot \EA(\mathcal{A}))\) rounds and using total communication \(\mathcal{O}(\min(n(V_A,V_B)\cdot \NA(\mathcal{A}), e(V_A,V_B)\cdot \EA(\mathcal{A}))\cdot (\log n +\log R_{\mathcal{A}}(n))\), such that each player computes the value of \(\mathcal{A}(\mathcal{I})\) at all nodes he or she controls. 
\end{lemma}

\begin{proof}

In protocol \(\mathcal{P}\), Alice and Bob simulate the rounds of algorithm \(\mathcal{A}\)  in all the nodes they control. The simulation run in phases. Each phase is used to simulate up to a certain number of rounds \(t\) of algorithm \(\mathcal{A}\), and takes two rounds of protocol \(\mathcal{P}\) (one round for Alice, and one round for Bob). By simulating \(\mathcal{A}\) up to \(t\) rounds, we mean that Alice and Bob know all the states of all the nodes they control, on every round up to round \(t\). 

In the first phase, players start simulating \(\mathcal{A}\) from the initial state. Let us suppose that both Alice and Bob have already executed \(p\geq 0\) phases, meaning that they had correctly simulated \(\mathcal{A}\) up to round \(t = t(p)\geq 0\). Let us explain phase \(p+1\) (see also Figure~\ref{fig:simulationphase}).

\begin{figure}[!h]
    \centering
    \tikzset{every picture/.style={line width=0.75pt}} 

\begin{tikzpicture}[x=0.75pt,y=0.75pt,yscale=-1,xscale=1]

\draw [color={rgb, 255:red, 155; green, 155; blue, 155 }  ,draw opacity=1 ][fill={rgb, 255:red, 155; green, 155; blue, 155 }  ,fill opacity=1 ] [dash pattern={on 0.84pt off 2.51pt}]  (165,191) -- (424.72,191) -- (590,191) ;
\draw [color={rgb, 255:red, 155; green, 155; blue, 155 }  ,draw opacity=1 ][fill={rgb, 255:red, 155; green, 155; blue, 155 }  ,fill opacity=1 ] [dash pattern={on 0.84pt off 2.51pt}]  (165,160) -- (424.72,160) -- (590,160) ;
\draw    (170,141) -- (170,227) ;
\draw [shift={(170,230)}, rotate = 270] [fill={rgb, 255:red, 0; green, 0; blue, 0 }  ][line width=0.08]  [draw opacity=0] (8.93,-4.29) -- (0,0) -- (8.93,4.29) -- cycle    ;
\draw    (213,160) -- (260,160) ;
\draw [shift={(210,160)}, rotate = 0] [fill={rgb, 255:red, 0; green, 0; blue, 0 }  ][line width=0.08]  [draw opacity=0] (5.36,-2.57) -- (0,0) -- (5.36,2.57) -- cycle    ;
\draw    (213,170) -- (260,170) ;
\draw [shift={(210,170)}, rotate = 0] [fill={rgb, 255:red, 0; green, 0; blue, 0 }  ][line width=0.08]  [draw opacity=0] (5.36,-2.57) -- (0,0) -- (5.36,2.57) -- cycle    ;
\draw    (210,200) -- (257,200) ;
\draw [shift={(260,200)}, rotate = 180] [fill={rgb, 255:red, 0; green, 0; blue, 0 }  ][line width=0.08]  [draw opacity=0] (5.36,-2.57) -- (0,0) -- (5.36,2.57) -- cycle    ;
\draw    (165,191) -- (175,191) ;
\draw [color={rgb, 255:red, 0; green, 0; blue, 0 }  ,draw opacity=1 ]   (360,191) -- (407,191) ;
\draw [shift={(410,191)}, rotate = 180] [fill={rgb, 255:red, 0; green, 0; blue, 0 }  ,fill opacity=1 ][line width=0.08]  [draw opacity=0] (5.36,-2.57) -- (0,0) -- (5.36,2.57) -- cycle    ;
\draw    (310,80) -- (310,230) ;
\draw [color={rgb, 255:red, 0; green, 0; blue, 0 }  ,draw opacity=1 ]   (503,160) -- (550,160) ;
\draw [shift={(500,160)}, rotate = 0] [fill={rgb, 255:red, 0; green, 0; blue, 0 }  ,fill opacity=1 ][line width=0.08]  [draw opacity=0] (5.36,-2.57) -- (0,0) -- (5.36,2.57) -- cycle    ;
\draw    (460,80) -- (460,240) ;
\draw    (165,160) -- (175,160) ;
\draw    (165,141) -- (175,141) ;
\draw    (210,150) -- (210,220) ;
\draw    (260,150) -- (260,220) ;
\draw    (360,150) -- (360,220) ;
\draw    (410,150) -- (410,220) ;
\draw    (500,150) -- (500,220) ;
\draw    (550,150) -- (550,220) ;

\draw (167,240.5) node  [font=\small] [align=left] {\begin{minipage}[lt]{27.54pt}\setlength\topsep{0pt}
rounds
\end{minipage}};
\draw (195,131.4) node [anchor=north west][inner sep=0.75pt]    {$V_{A}$};
\draw (261,131.4) node [anchor=north west][inner sep=0.75pt]    {$V_{B}$};
\draw (145,181.4) node [anchor=north west][inner sep=0.75pt]    {$r_{a}$};
\draw (144,152.4) node [anchor=north west][inner sep=0.75pt]    {$r_{b}$};
\draw (341,131.4) node [anchor=north west][inner sep=0.75pt]    {$V_{A}$};
\draw (411,131.4) node [anchor=north west][inner sep=0.75pt]    {$V_{B}$};
\draw (481,131.4) node [anchor=north west][inner sep=0.75pt]    {$V_{A}$};
\draw (551,131.4) node [anchor=north west][inner sep=0.75pt]    {$V_{B}$};
\draw (324,82) node [anchor=north west][inner sep=0.75pt]  [font=\small] [align=left] {\begin{minipage}[lt]{86.85pt}\setlength\topsep{0pt}
\begin{center}
Oblivious simulation\\of Alice
\end{center}

\end{minipage}};
\draw (146,125.4) node [anchor=north west][inner sep=0.75pt]    {$t$};
\draw (471,82) node [anchor=north west][inner sep=0.75pt]  [font=\small] [align=left] {\begin{minipage}[lt]{86.85pt}\setlength\topsep{0pt}
\begin{center}
Oblivious simulation\\of Bob
\end{center}

\end{minipage}};
\draw (191,83) node [anchor=north west][inner sep=0.75pt]  [font=\small] [align=left] {\begin{minipage}[lt]{54.95pt}\setlength\topsep{0pt}
\begin{center}
Transcript of\\algorithm $\displaystyle \mathcal{A}$
\end{center}

\end{minipage}};

\end{tikzpicture}
    \caption{Illustration of one phase of the simulation protocol. Assuming that the players agree on the simulation of  algorithm \(\mathcal{A}\) up to round \(t\), each player runs an oblivious simulation at the nodes they control. In the example of the figure, the next message corresponds to a node controlled by Bob, who sends a message to a node in \(V_A\) at round \(r_b\). The oblivious simulation of Alice is not aware of this message, and incorrectly considers that a message is sent from \(V_A\) to \(V_B\) at round \(r_a > r_b\). Using the communication rounds in this phase, the players agree that the message of Bob was correct. Thus  the simulation is correct up to round \(r_b\), for both players.   }
    \label{fig:simulationphase}
\end{figure}
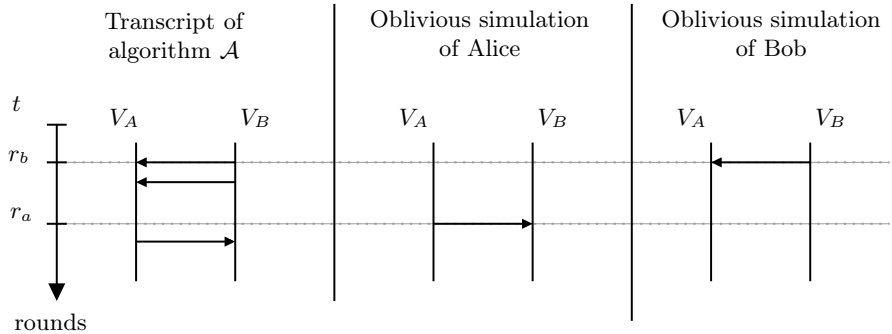

Starting from round \(t\), Alice runs an \emph{oblivious simulation} of algorithm \(\mathcal{A}\) over all nodes that she controls. By oblivious, we mean that Alice assumes that no node of \(V_B\) communicates a message to a node in \(V_A\) in any round at least \(t\). The oblivious simulation of Alice stops in one of the following two possible scenarios:

\begin{itemize}
\item[(1)] All nodes that she controls either terminate or enter into a passive state that quits only on an incoming message from \(V_B\). 
\item[(2)] The simulation reaches a round \(r_a\) where a message is sent from a node in \(V_A\) to a node in \(V_B\).
\end{itemize}

At the same time, Bob runs and oblivious simulation of \(\mathcal{A}\) starting from round \(t\) (i.e. assuming that no node of \(V_A\) sends a message to a node in \(V_B\) in any round at least \(t\)). The oblivious simulation of Bob stops in one of the same two  scenarios analogous to the ones  above. In this case, we call \(r_b\) the round reached by Bob in his version of scenario~(2). 

At the beginning of a phase, it is the turn of Alice to speak.  Once the oblivious simulation of Alice stops, she is ready to send a message to Bob. If the simulation stops in the  scenario (1), Alice sends a message "\emph{scenario 1}" to Bob. Otherwise, Alice sends \(r_a\) together with all the messages sent from nodes in \(V_A\) to nodes in \(V_B\) at round \(r_a\), to Bob. When Bob receives the message from Alice, one of the following situations holds:\\

\noindent{Case 1:} the oblivious simulation of both Alice and Bob stopped in the first scenario. In this case, since \(\mathcal{A}\) is correct, there are no deadlocks. Therefore, all vertices of \(G\) reached a terminal state, meaning that the oblivious simulation of both players was in fact a real simulation of \(\mathcal{A}\), and the obtained states are the output states. Therefore, Bob sends a message to Alice indicating that the simulation is finished, and indeed Alice and Bob have correctly computed the output of \(\mathcal{A}\) for all the nodes they control. \\

\noindent{Case 2:} the oblivious simulation of Alice stopped in scenario (1), and the one of Bob stopped in the scenario (2). In this case, Bob infers that his oblivious simulation was correct. He sends  \(r_b\) and all the messages communicated in round \(r_b\) through the cut to Alice. When Alice receives the message of Bob, she updates the state of the nodes she controls up to round \(r_b\). It follows that both players have correctly simulated algorithm \(\mathcal{A}\) up to round \(r_b > t\).\\ 

\noindent{Case 3:} the oblivious simulation of Alice stopped in scenario (2), and the one of Bob stopped in scenario (1). In this case, Bob infres that the simulation of Alice was correct up to round \(r_a\). He sends a message to Alice indicating that she has correctly simulated \(\mathcal{A}\) up to round \(r_a\), and he updates the states of all the nodes he controls up to round \(r_a\). It follows that both players have correctly simulated  \(\mathcal{A}\) up to round \(r_a > t\). \\

\noindent{Case 4:} the oblivious simulation of both players stopped in scenario (2), and \(r_a > r_b\). Bob infers that his oblivious simulation was correct up to \(r_b\), and that the one of Alice was not correct after round \(r_b\). Then, the players act in the same way as   described in Case 2. Thus, both players have correctly simulated \(\mathcal{A}\) up to round~\(r_b\).\\

\noindent{Case 5:} the oblivious simulation of both players stopped in scenario (2), and \(r_b > r_a\). Bob infers that his oblivious simulation was incorrect after round \(r_a\), and that the one of Alice was correct up to round \(r_a\). Then, the players act in the same way as described in Case 3. Thus, both players have correctly simulated \(\mathcal{A}\) up to round~\(r_a\). \\

\noindent{Case 6:} the oblivious simulation of both players stopped in scenario (2), and \(r_b = r_a\). Bob assumes that both oblivious simulations were correct. He sends \(r_b\) together with all the messages communicated from his nodes at round \(r_b\) through the cut. Then, he updates the states of all the nodes he controls up to round \(r_b\). When Alice receives the message from Bob, she updates the states of the nodes she controls up to round \(r_b\).
It follows that both players have correctly simulated  \(\mathcal{A}\) up to round \(r_b > t\). \\

Observe that, except when the algorithm terminates, on each phase of the protocol, at least one node controlled by Alice or Bob is activated. Since the number of rounds of \(\mathcal{P}\) is twice the number of phases, we deduce that the total number of rounds is at most
\[ 2\cdot \min(n(V_A,V_B)\cdot\NA(\mathcal{A}), e(V_A,V_B)\cdot \EA(\mathcal{A})).  \]
Moreover, on each round of $\mathcal{P}$,  the players communicate \(O((\log(R_{\mathcal{A}}(n)) + \log n) \cdot e(V_A,V_B))\) bits. As a consequence, the total communication cost of $\mathcal{P}$ is \[O((\log(R_{\mathcal{A}}(n)) + \log n )\cdot e(V_A,V_B)) \cdot \min(n(V_A,V_B)\cdot \NA(\mathcal{A}), e(V_A,V_B)\cdot \EA(\mathcal{A})) ),\]
which completes the proof. \qed \end{proof}

We use the simulation lemma to show that there are problems that cannot be solved by a frugal algorithm in a polynomial number of rounds. In problem \CFF{}, all nodes of the input graph $G$ must accept if $G$ has no cycle of 4 vertices, and at least one node must reject if such a cycle exists. Observe that this problem is expressible in first-order logic, in particular it has en edge-frugal algorithm with a polynomial number of rounds in graphs of bounded degree~\cite{GrumbachW09}. We show that, in graphs of unbounded degree, this does not hold anymore.
We shall also consider problem \SYM{}, where the input is a graph $G$ with $2n$ nodes indexed from $1$ to $2n$, and with a unique edge  $\{1,n+1\}$ between $G_A = G[\{1,\dots,n\}]$ and $G_B = G[\{n+1,\dots,2n\}]$. Our lower bounds holds even if  every node is identified by its index. All nodes must output \emph{accept} if the function $f:\{1,\dots,n\} \to \{n+1,\dots,2n\}$ defined by $f(x)=x+n$ is an isomorphism from $G_A$ to $G_B$, otherwise at least one node must output \emph{reject}.

 The proof of the following theorem is based on classic reductions from communication complexity problems \textsc{Equality} and \textsc{Set Disjointness} (see, e.g., \cite{KushilevitzNisan}), combined with Lemma~\ref{lem:simu}. 

\begin{theorem}\label{th:limits}
Any CONGEST algorithm solving \SYM{} (resp., \CFF{}) in polynomially many rounds has node-activation and edge-activation at least $\Omega\left(\frac{n^2}{\log n}\right)$ (resp., $\Omega\left(\frac{\sqrt{n}}{\log n}\right)$).
\end{theorem}

\begin{proof}
In problem \textsc{Equality}, two players Alice and Bob have a boolean vector of size $k$, $x_A$ for Alice and $x_B$ for Bob. Their goal is to answer \emph{true} if $x_A = x_B$,  and \emph{false} otherwise. The communication complexity of this problem is known to be~$\Theta(k)$~\cite{KushilevitzNisan}.
Let $k = n^2$. We can interpret $x_A$ and $x_B$ as the adjacency matrix of two graphs $G_A$ and $G_B$ in an instance of $\SYM$. It is a mere technicality to "shift" $G_B$ as if its vertices were indexed from $1$ to~$n$, such that \SYM{} is true for $G$ iff $x_A = x_B$. Moreover, Alice can construct $G_A$ from its input $x_A$, and Bob can construct $G_B$ from $x_B$. Both can simulate the unique edge joining the two graphs in $G$. Therefore, by Lemma~\ref{lem:simu} applied to $G$, if Alice controls the vertices of $G_A$, and Bob controls the vertices of $G_B$, then any CONGEST algorithm $\mathcal{A}$ solving \SYM{} in polynomially many rounds yields a two-party protocol for \textsc{Equality} on $n^2$ bits. Since graphs $G_A$ and $G_B$ are linked by a unique edge, the total communication of the protocol is $O(\EA(\mathcal{A}) \cdot \log n)$ and  $O(\NA(\mathcal{A})\cdot \log n)$. The result follows.

\medskip 

In  \textsc{Set Disjointness}, each of the two players Alice and Bob has a Boolean vector of size $k$, $x_A$ for Alice, and $x_B$ for Bob. Their goal is to answer \emph{true} if there is no index $i \in [k]$ such that both $x_A[i]$ and $x_B[i]$ are true (in which case, $x_A$ and $x_B$ are disjoint),  and \emph{false} otherwise. The communication complexity of this problem is known to be~$\Theta(k)$~\cite{KushilevitzNisan}.
We use the technique in~\cite{DruckerKO14} to construct an  instance $G$ for $C_4$ freeness, with a small cut, from two Boolean vectors $x_A, x_B$ of size $k = \Theta(n^{3/2})$. Consider a $C_4$-free $n$-vertex graph~$H$ with a maximum number of edges. Such a graph has $k = \Theta(n^{3/2})$ edges, as recalled in~\cite{DruckerKO14}. We can consider the edges $E(H)$ as indexed from $1$ to~$k$, and $V(H)$ as $[n]$. Let now $x_A$ and $x_B$ be two Boolean vectors of size~$k$. These vectors can be interpreted as edge subsets $E(x_A)$ and $E(x_B)$ of~$H$, in the sense that the edge indexed $i$ in $E(H)$ appears in $E(x_A)$ (resp. $E(x_B)$) iff $x_A[i]$ (resp. $x_B[i]$) is true. Graph $G$ is constructed to have $2n$ vertices, formed by two sub-graphs $G_A = G[\{1,\dots, n\}]$ and $G_B = G[\{n+1,\dots, 2n\}]$. 
The edges of $E(G_A)$ are exactly the ones of $E(x_A)$. Similarly, the edges of $E(G_B)$ correspond to $E(x_A)$, modulo the fact that the vertex indexes are shifted by~$n$, i.e., for each edge $\{u,v\} \in E(x_B)$, we add edge  $\{u+n,v+n\}$ to~$G_B$. Moreover we add a perfect matching to $G$, between $V(G_A)$ and $V(G_B)$, by adding all edges $\{i,i+n\}$, for all $i \in [n]$. Note that $G$ is $C_4$-free if and only if vectors $x_A$ and $x_B$ are disjoint. Indeed, since $G_A, G_B$ are isomorphic to sub-graphs of $H$, they are $C_4$-free. Thus any $C_4$ of $G$ must contain two vertices in $G_A$ and two in $G_B$, in which case the corresponding edges in $G_A$ and $G_B$ designate the same bit of $x_A$ and $x_B$ respectively.
Moreover Alice and Bob can construct $G_A$ and $G_B$, as well as the edges in the matching, from their respective inputs $x_A$ and~$x_B$.  Therefore, thanks to Lemma~\ref{lem:simu}, a CONGEST algorithm $\mathcal{A}$ for \CFF{} running in a polynomial number of rounds can be used to design a protocol $\mathcal{P}$ solving \textsc{Set Disjointness} on $k = \Theta(n^{3/2})$ bits, where Alice controls $V(G_A)$ and Bob controls $V(G_B)$. The communication complexity of the protocol is $O(\EA(\mathcal{A})\cdot n\cdot \log n)$, and $O(\NA(\mathcal{A}) \cdot n\cdot \log n)$, since the cut between $G_A$ and $G_B$ is a matching. The result follows. 
\qed    
\end{proof}

\section{Node versus Edge Activation}

In this section we exhibit a problem that admits an edge-frugal CONGEST algorithm running in a polynomial number of rounds,  for which any algorithm running in a polynomial number of rounds has large node-activation complexity.

We proceed by reduction from a two-party communication complexity problem. However, unlike the previous section, we are now also interested in the number of rounds of the two-party protocols. We consider protocols in which the two players Alice and Bob do not communicate simultaneously. For such a protocol~$\mathcal{P}$, a \emph{round} is defined as a maximal contiguous sequence of messages emitted by a same player. We denote by $R(\mathcal{P})$ the number of rounds of~$\mathcal{P}$. 

Let \(G\) be a graph, and \(S\) be a subset of nodes of~$G$. We denote by \(\partial S\) the number of vertices in \(S\) with a neighbor in \(V\setminus S\). 

\begin{lemma}[Round-Efficient Simulation lemma]\label{lem:simu2}
    Let  \(\mathcal{A}\) be an algorithm in the CONGEST model, let \(\mathcal{I} = (G,\id,x)\) be an input for~$\mathcal{A}$, and let \(V_A, V_B\) be a partition of~\(V(G)\).  Let us assume that Alice controls all the nodes in~\(V_A\), and Bob controls all the nodes in \(V_B\), and both players know the value of \(\NA(\mathcal{A})\).  Then, there exists a communication protocol \(\mathcal{P}\) between Alice and Bob such that, in at most \(\min(\partial V_A, \partial V_B)\cdot \NA(\mathcal{A})\) rounds, and using total communication \(O\Big(\big((\partial(V_A) + \partial(V_B)) \cdot \NA(\mathcal{A})\big)^2 \cdot (\log n + \log R_{\mathcal{A}}(n))\Big)  \) bits, each player computes the value of \(\mathcal{A}(\mathcal{I})\) at all the nodes he or she controls. 
\end{lemma}

\begin{proof}
In protocol \(\mathcal{P}\), Alice and Bob simulate the rounds of algorithm \(\mathcal{A}\)  at all the nodes each player controls. Without loss of generality, we assume that algorithm \(\mathcal{A}\) satisfies that the nodes send messages at different rounds, by merely multiplying by \(N\) the number of rounds. 

Initially,  Alice runs an oblivious simulation of \(\mathcal{A}\) that stops in one of the following three cases: 
\begin{enumerate}
\item Every node in \(V_A\)  has terminated; 
\item Every node in \(V_A\) entered into the passive state that it may leave only after having received a message from a node in~\(V_B\) (this corresponds to what we called the ``first scenario'' in the proof of Lemma~\ref{lem:simu}); 
\item The number of rounds \(R_{\mathcal{A}}(n)\) is reached. 
\end{enumerate}

Then, Alice sends to Bob the integer \(t_1 = 0\), and the set \(M^1_A\) of all messages sent from nodes in \(V_A\) to nodes in \(V_B\) in the communication rounds that she simulated, together with their corresponding timestamps. If the number of messages communicated by Alice exceeds \(\NA(\mathcal{A})\cdot \partial A\), we trim the list up to this threshold.

Let us suppose that the protocol \(\mathcal{P}\) has run for \(p\) rounds, and let us assume that  it is the turn of Bob to speak at round \(p+1\) --- the case where Alice speaks at round \(p+1\) can be treated in the same way. Moreover, we assume that \(\mathcal{P}\) satisfies the following two conditions:
\begin{enumerate}
\item At round \(p\), Alice sents an integer \(t_p\geq 0\), and a list of timestamped messages \(M^p_A\) corresponding to messages sent from nodes in \(V_A\) to nodes in \(V_B\) in an oblivious simulation of \(\mathcal{A}\) starting from a round \(t_p\). 
\item Bob had correctly simulated \(\mathcal{A}\) at all the nodes he controls, up to round \(t_p\). 
\end{enumerate}

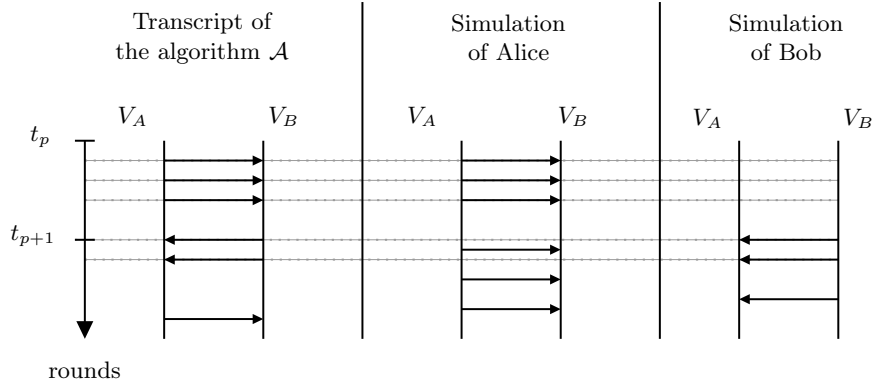
\begin{figure}[!h]
    \centering
    \tikzset{every picture/.style={line width=0.75pt}} 

\begin{tikzpicture}[x=0.75pt,y=0.75pt,yscale=-1,xscale=1]

\draw [color={rgb, 255:red, 155; green, 155; blue, 155 }  ,draw opacity=1 ][fill={rgb, 255:red, 155; green, 155; blue, 155 }  ,fill opacity=1 ] [dash pattern={on 0.84pt off 2.51pt}]  (165,200) -- (424.72,200) -- (550,200) ;
\draw [color={rgb, 255:red, 155; green, 155; blue, 155 }  ,draw opacity=1 ][fill={rgb, 255:red, 155; green, 155; blue, 155 }  ,fill opacity=1 ] [dash pattern={on 0.84pt off 2.51pt}]  (170,210) -- (427.36,210) -- (550,210) ;
\draw [color={rgb, 255:red, 155; green, 155; blue, 155 }  ,draw opacity=1 ][fill={rgb, 255:red, 155; green, 155; blue, 155 }  ,fill opacity=1 ] [dash pattern={on 0.84pt off 2.51pt}]  (170,160) -- (427.36,160) -- (550,160) ;
\draw [color={rgb, 255:red, 155; green, 155; blue, 155 }  ,draw opacity=1 ][fill={rgb, 255:red, 155; green, 155; blue, 155 }  ,fill opacity=1 ] [dash pattern={on 0.84pt off 2.51pt}]  (170,170) -- (427.36,170) -- (550,170) ;
\draw [color={rgb, 255:red, 155; green, 155; blue, 155 }  ,draw opacity=1 ][fill={rgb, 255:red, 155; green, 155; blue, 155 }  ,fill opacity=1 ] [dash pattern={on 0.84pt off 2.51pt}]  (170,180) -- (427.36,180) -- (550,180) ;
\draw    (170,150) -- (170,247) ;
\draw [shift={(170,250)}, rotate = 270] [fill={rgb, 255:red, 0; green, 0; blue, 0 }  ][line width=0.08]  [draw opacity=0] (8.93,-4.29) -- (0,0) -- (8.93,4.29) -- cycle    ;
\draw    (210,160) -- (257,160) ;
\draw [shift={(260,160)}, rotate = 180] [fill={rgb, 255:red, 0; green, 0; blue, 0 }  ][line width=0.08]  [draw opacity=0] (5.36,-2.57) -- (0,0) -- (5.36,2.57) -- cycle    ;
\draw    (210,170) -- (257,170) ;
\draw [shift={(260,170)}, rotate = 180] [fill={rgb, 255:red, 0; green, 0; blue, 0 }  ][line width=0.08]  [draw opacity=0] (5.36,-2.57) -- (0,0) -- (5.36,2.57) -- cycle    ;
\draw    (213,200) -- (260,200) ;
\draw [shift={(210,200)}, rotate = 0] [fill={rgb, 255:red, 0; green, 0; blue, 0 }  ][line width=0.08]  [draw opacity=0] (5.36,-2.57) -- (0,0) -- (5.36,2.57) -- cycle    ;
\draw    (165,200) -- (175,200) ;
\draw [color={rgb, 255:red, 0; green, 0; blue, 0 }  ,draw opacity=1 ]   (503,210) -- (550,210) ;
\draw [shift={(500,210)}, rotate = 0] [fill={rgb, 255:red, 0; green, 0; blue, 0 }  ,fill opacity=1 ][line width=0.08]  [draw opacity=0] (5.36,-2.57) -- (0,0) -- (5.36,2.57) -- cycle    ;
\draw    (310,80) -- (310,250) ;
\draw    (460,80) -- (460,250) ;
\draw    (165,150) -- (175,150) ;
\draw    (210,150) -- (210,250) ;
\draw    (260,150) -- (260,250) ;
\draw    (500,150) -- (500,250) ;
\draw    (550,150) -- (550,250) ;
\draw    (213,210) -- (260,210) ;
\draw [shift={(210,210)}, rotate = 0] [fill={rgb, 255:red, 0; green, 0; blue, 0 }  ][line width=0.08]  [draw opacity=0] (5.36,-2.57) -- (0,0) -- (5.36,2.57) -- cycle    ;
\draw    (210,180) -- (257,180) ;
\draw [shift={(260,180)}, rotate = 180] [fill={rgb, 255:red, 0; green, 0; blue, 0 }  ][line width=0.08]  [draw opacity=0] (5.36,-2.57) -- (0,0) -- (5.36,2.57) -- cycle    ;
\draw [color={rgb, 255:red, 0; green, 0; blue, 0 }  ,draw opacity=1 ]   (503,200) -- (550,200) ;
\draw [shift={(500,200)}, rotate = 0] [fill={rgb, 255:red, 0; green, 0; blue, 0 }  ,fill opacity=1 ][line width=0.08]  [draw opacity=0] (5.36,-2.57) -- (0,0) -- (5.36,2.57) -- cycle    ;
\draw [color={rgb, 255:red, 0; green, 0; blue, 0 }  ,draw opacity=1 ]   (503,230) -- (550,230) ;
\draw [shift={(500,230)}, rotate = 0] [fill={rgb, 255:red, 0; green, 0; blue, 0 }  ,fill opacity=1 ][line width=0.08]  [draw opacity=0] (5.36,-2.57) -- (0,0) -- (5.36,2.57) -- cycle    ;
\draw    (210,240) -- (257,240) ;
\draw [shift={(260,240)}, rotate = 180] [fill={rgb, 255:red, 0; green, 0; blue, 0 }  ][line width=0.08]  [draw opacity=0] (5.36,-2.57) -- (0,0) -- (5.36,2.57) -- cycle    ;
\draw    (360,170) -- (407,170) ;
\draw [shift={(410,170)}, rotate = 180] [fill={rgb, 255:red, 0; green, 0; blue, 0 }  ][line width=0.08]  [draw opacity=0] (5.36,-2.57) -- (0,0) -- (5.36,2.57) -- cycle    ;
\draw    (360,205) -- (407,205) ;
\draw [shift={(410,205)}, rotate = 180] [fill={rgb, 255:red, 0; green, 0; blue, 0 }  ][line width=0.08]  [draw opacity=0] (5.36,-2.57) -- (0,0) -- (5.36,2.57) -- cycle    ;
\draw    (360,150) -- (360,250) ;
\draw    (410,150) -- (410,250) ;
\draw    (360,220) -- (407,220) ;
\draw [shift={(410,220)}, rotate = 180] [fill={rgb, 255:red, 0; green, 0; blue, 0 }  ][line width=0.08]  [draw opacity=0] (5.36,-2.57) -- (0,0) -- (5.36,2.57) -- cycle    ;
\draw    (360,180) -- (407,180) ;
\draw [shift={(410,180)}, rotate = 180] [fill={rgb, 255:red, 0; green, 0; blue, 0 }  ][line width=0.08]  [draw opacity=0] (5.36,-2.57) -- (0,0) -- (5.36,2.57) -- cycle    ;
\draw    (360,235) -- (407,235) ;
\draw [shift={(410,235)}, rotate = 180] [fill={rgb, 255:red, 0; green, 0; blue, 0 }  ][line width=0.08]  [draw opacity=0] (5.36,-2.57) -- (0,0) -- (5.36,2.57) -- cycle    ;
\draw    (360,160) -- (407,160) ;
\draw [shift={(410,160)}, rotate = 180] [fill={rgb, 255:red, 0; green, 0; blue, 0 }  ][line width=0.08]  [draw opacity=0] (5.36,-2.57) -- (0,0) -- (5.36,2.57) -- cycle    ;

\draw (170,265.5) node  [font=\small] [align=left] {\begin{minipage}[lt]{27.54pt}\setlength\topsep{0pt}
rounds
\end{minipage}};
\draw (185,131.4) node [anchor=north west][inner sep=0.75pt]    {$V_{A}$};
\draw (261,131.4) node [anchor=north west][inner sep=0.75pt]    {$V_{B}$};
\draw (475,132.4) node [anchor=north west][inner sep=0.75pt]    {$V_{A}$};
\draw (551,132.4) node [anchor=north west][inner sep=0.75pt]    {$V_{B}$};
\draw (331,131.4) node [anchor=north west][inner sep=0.75pt]    {$V_{A}$};
\draw (407,131.4) node [anchor=north west][inner sep=0.75pt]    {$V_{B}$};
\draw (181,84) node [anchor=north west][inner sep=0.75pt]  [font=\small] [align=left] {\begin{minipage}[lt]{69.93pt}\setlength\topsep{0pt}
\begin{center}
Transcript of\\the algorithm $\displaystyle \mathcal{A}$
\end{center}

\end{minipage}};
\draw (351,85) node [anchor=north west][inner sep=0.75pt]  [font=\small] [align=left] {\begin{minipage}[lt]{46.89pt}\setlength\topsep{0pt}
\begin{center}
Simulation\\of Alice
\end{center}

\end{minipage}};
\draw (141,141.4) node [anchor=north west][inner sep=0.75pt]    {$t_{p}$};
\draw (491,85) node [anchor=north west][inner sep=0.75pt]  [font=\small] [align=left] {\begin{minipage}[lt]{46.89pt}\setlength\topsep{0pt}
\begin{center}
Simulation\\of Bob
\end{center}

\end{minipage}};
\draw (131,191.4) node [anchor=north west][inner sep=0.75pt]    {$t_{p+1}$};

\end{tikzpicture}
    \caption{Illustration of the round-efficient simulation protocol for algorithm \(\mathcal{A}\). After round \(p\), Alice has correctly simulated the algorithm up to round \(t_p\). It is the turn of Bob to speak in round \(p+1\). In round \(p\), Alice sent to Bob the set of messages \(M^p_A\), obtained from an oblivious simulation of \(\mathcal{A}\) starting from \(t_p\). Only the first three messages are correct, since at round \(t_{p+1}\) Bob communicates a message to Alice. Then, Bob runs an oblivious simulation of \(\mathcal{A}\) starting from \(t_{p+1}\), and communicates all the messages sent from nodes \(V_B\) to nodes in \(V_A\).  In this case the two first messages are correct. }
    \label{fig:efficientsimulationround}
\end{figure}

We now describe round \(p+1\) (see also Figure~\ref{fig:efficientsimulationround}). 
Bob initiates a simulation of \(\mathcal{A}\) at all the nodes he controls. However, this simulation is \emph{not} oblivious. Specifically, Bob simulates \(\mathcal{A}\) from round \(t_p\) taking into account all the messages sent from nodes in \(V_A\) to nodes in \(V_B\), as listed in the messages~\(M^{p}_A\). The simulation stops when Bob reaches a round \(t_{p+1}>t_p\) at which a node in \(V_B\) sends a message to a node in \(V_A\). Observe that, up to round \(t_{p+1}\), the oblivious simulation of Alice was correct. At this point, Bob initiates an oblivious simulation of \(\mathcal{A}\) at all the nodes he controls, starting from \(t_{p+1}\). Finally, Bob sends to Alice \(t_{p+1}\), and the list \(M^{p+1}_B\) of all timestamped messages sent from nodes in \(V_B\) to nodes in \(V_A\) resulting from the oblivious simulation of the nodes he controls during rounds at least \(t_{p+1}\). Using this information, Alice infers that her simulation was correct up to round \(t_{p+1}\), and she starts the next round for protocol~$\mathcal{P}$. 

The simulation carries on until one of the two players runs an oblivious simulation in which all the nodes he or she controls terminate, and no messages were sent through the cut in at any intermediate round. In this case, this player sends a message "\emph{finish}" to the other player, and both infer that their current simulations are correct. As  a consequence, each player has correctly computed the output of \(\mathcal{A}\) at all the nodes he or she controls.

At every communication round during which Alice speaks, at least one vertex of \(V_A\) which has a neighbor in \(V_B\) is activated. Therefore, the number of rounds of Alice is at most \(\partial V_A \cdot \NA(\mathcal{A})\). By the same argument, we have that the number of rounds of Bob is at most \(\partial V_B \cdot \NA(\mathcal{A})\). It follows that 
\[
R(\mathcal{P}) = \min(\partial V_A,\partial V_B) \cdot \NA(\mathcal{A}).
\]
At each communication round, Alice sends at most \(\partial(V_A) \cdot \NA(\mathcal{A})\) timestamped messages, which can be encoded using \(O\big(\partial(V_A) \cdot \NA(\mathcal{A}) \cdot (\log n + \log R_{\mathcal{A}}(n))\big) \) bits. Similarly, Bob sends \(O\big (\partial(V_B) \cdot \NA(\mathcal{A}) \cdot (\log n + \log R_{\mathcal{A}}(n)) \big) \) bits. It follows that 
\[
C(\mathcal{P}) = O\Big(\big((\partial(V_A) + \partial(V_B)) \cdot \NA(\mathcal{A})\big)^2 \cdot (\log n + \log R_{\mathcal{A}}(n))\Big),
\]
which completes the proof. 
\qed
\end{proof}

In order to separate the node-activation complexity from the edge-activation complexity, we consider a problem called \DFPC, and we show that this problem can be solved by an edge-frugal CONGEST algorithm running in $O(\mbox{poly}(n))$ rounds, whereas the node-activation complexity of any algorithm running in $O(\mbox{poly}(n))$ rounds for this problem is $\Omega(\Delta)$, for any $\Delta \in O(\UpB)$. The lower bound is proved thanks to the Round-Efficient Simulation Lemma (Lemma~\ref{lem:simu}), by reduction from the two-party communication complexity problem \PC{}, for which too few rounds imply large communication complexity~\cite{NisanW93}.

In the \DFPC, each node $v$ of the graph is given as input its index $\mbox{DFS}(v)\in [n]$ in a depth-first search ordering (as usual we denote $[n]=\{1,\dots,n\}$). Moreover the vertex indexed~$i$ is given a function $f_i:[n] \to [n]$, and the root (i.e., the node indexed~1) is given a value $x \in [n]$ as part of its input. The goal is to compute the value of $f_n \circ f_{n-1} \circ \dots \circ f_1(x)$ at the root. 

\begin{lemma}\label{lem:algoforDFPC}
There exists an edge-frugal CONGEST algorithm for problem \DFPC, with polynomial number of rounds.
\end{lemma}


\begin{proof}
The lemma is established using an algorithm that essentially traverses the DFS tree encoded by the indices of the nodes, and performs the due partial computation of the function at every node, that is, the node with index~$i$ computes $f_i \circ f_{i-1} \dots f_1 (x)$, and forwards the result to the node with index~$i+1$.

At round 1, each node $v$ transmits its depth-first search index $\mbox{DFS}(v)$ to its neighbors. Therefore, after this round, every node knows its parent, and its children in the DFS tree. Then the algorithm merely forwards messages of type $m(i) = f_i \circ f_{i-1} \dots f_1 (x)$, corresponding to iterated computations for increasing values~$i$, along the DFS tree, using the DFS ordering. That is, for any node $v$, let $\mbox{MaxDFS}(v)$ denote the maximum DFS index appearing in the subtree of the DFS tree rooted at $v$. We will not explicitly compute this quantity but it will ease the notations. At some round, vertex $v$ of DFS index $i$ will receive a message $m(i-1)$ from its parent (of index $i-1$). Then node $v$ will be in charge of computing message $m(\mbox{MaxDFS}(v))$, by ``calling'' its children in the tree, and sending this message back to its parent. In this process, each edge in the subtree rooted at $v$ is activated twice.

The vertex of DFS index~1 initiates the process at round~2, sending $f_1(x)$ to its child of DFS index~$2$. Any other node~$v$ waits until it receives a message from its parent, at a round that we denote~$r(v)$. This  message is precisely $m(i-1) = f_{i-1} \circ f_{i-2} \dots f_1 (x)$, for $i = \mbox{DFS}(v)$. Then $v$ computes message 
$m(i) = f_{i} \circ f_{i-1} \dots f_1 (x)$ using its local function $f_i$. If it has no children, then it sends  this message $m(i)$ to its parent at round $r(v)+1$. Assume now that $v$ has $j$ children in the DFS tree, denoted  $u_1,u_2,\dots,u_j$, sorted by increasing DFS index. Observe that, by definition of DFS trees, $\mbox{DFS}(u_k) = \mbox{MaxDFS}(u_{k-1})+1$ for each $k \in \{2,\dots,j\}$. 
Node~$v$ will be activated $j$ times, once for each edge $\{v,u_k\}$, $1 \leq k \leq j$, as follows. 
At round $r(v)+1$ (right after receiving the message from its parent), $v$~sends message $m(i)$ to its child $u_1$, then it awaits until round $r^1(v)$ when it gets back a message from $u_1$. 

The process is repeated for $k=2,\dots,j$: at round $r^{k-1}(v)+1$, node $v$ sends the message $m(\mbox{DFS}(u_{k})-1)$ received from $u_{k-1}$ to $u_k$, and waits until it gets back  a message from $u_k$, at round $r^k(v)$. Note that if $k<j$ then this message is $m(\mbox{DFS}(u_{k+1})-1)$, and if $k=j$ then this message is $m(\mbox{MaxDFS}(v))$. At round $r^j(v)+1$, after having received messages from all its children, $v$~backtracks message $m(\mbox{MaxDFS}(v))$ to its parent. If $v$ is the root, then the process stops.

The process terminates in $O(n)$ rounds, and, except for the first round,  every edge of the DFS tree is activated twice: first, going downwards, from the root towards the leaves, and, second, going upwards. At the end, the root obtains the requested message $m(n) = f_{n} \circ f_{n-1} \dots f_1 (x)$.
\qed
\end{proof}

\medbreak

Let us recall the \PC{} problem as defined in~\cite{NisanW93}.  
Alice is given a function $f_A : [n] \to [n]$, and a number $x_0 \in [n]$. Bob is given function $f_B : [n] \to [n]$. Both players have a parameter $k \in [n]$. Note that the size of the input given to each player is $\Theta(n \log n)$ bits.
The goal is to compute $(f_A \circ f_B)^k(x_0)$, i.e., $k$ successive iterations of $f_A \circ f_B$ applied to~$x_0$. We give a slightly simplified version of the result in~\cite{NisanW93}.

\begin{lemma}[Nissan and Wigderson \cite{NisanW93}]
\label{lem:PointerChasing}
Any two-party protocol for \PC{} using less than $2k$ rounds has communication complexity $\Omega(n - k \log n)$.
\end{lemma}

We have now all ingredients for proving the main result of this section. 

\begin{theorem}\label{thm:NodeEdgeSep}
For every $\Delta \in O\left(\UpB\right)$, every CONGEST algorithm solving  \DFPC{} in graphs of maximum degree $\Delta$ with polynomialy many rounds has node-activation complexity~$\Omega(\Delta)$.
\end{theorem}

\begin{proof}
Let $k$ be the parameter of \PC{} that will be fixed later. The lower bound is established for this specific parameter~$k$. Let us consider an arbitrary instance of \textsc{Pointer Chasing} $f_A, f_B:[n] \to [n]$, and $x_0 \in [n]$, with parameter~$k$. 
We reduce that instance to a particular instance of \textsc{Depth First Pointer Chasing} (see Fig.~\ref{fig:DFPC}). 

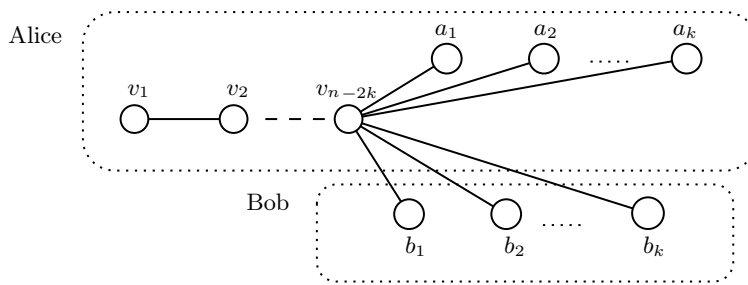
\begin{figure}[!h]
    \centering
    \tikzset{every picture/.style={line width=0.75pt}} 

\begin{tikzpicture}[x=0.75pt,y=0.75pt,yscale=-1,xscale=1]

\draw    (447.5,103.5) -- (398,134) ;
\draw    (496.5,103.5) -- (398,134) ;
\draw    (568.5,103.5) -- (398,134) ;
\draw    (398,134) -- (428.5,182) ;
\draw    (477.5,182) -- (398,134) ;
\draw    (549,181.5) -- (398,134) ;
\draw  [fill={rgb, 255:red, 255; green, 255; blue, 255 }  ,fill opacity=1 ] (421,182) .. controls (421,177.86) and (424.36,174.5) .. (428.5,174.5) .. controls (432.64,174.5) and (436,177.86) .. (436,182) .. controls (436,186.14) and (432.64,189.5) .. (428.5,189.5) .. controls (424.36,189.5) and (421,186.14) .. (421,182) -- cycle ;
\draw    (340,134) -- (290,134) ;
\draw  [dash pattern={on 4.5pt off 4.5pt}]  (398,134) -- (352,134) ;
\draw  [fill={rgb, 255:red, 255; green, 255; blue, 255 }  ,fill opacity=1 ] (391,134) .. controls (391,130.13) and (394.13,127) .. (398,127) .. controls (401.87,127) and (405,130.13) .. (405,134) .. controls (405,137.87) and (401.87,141) .. (398,141) .. controls (394.13,141) and (391,137.87) .. (391,134) -- cycle ;
\draw  [fill={rgb, 255:red, 255; green, 255; blue, 255 }  ,fill opacity=1 ] (333,134) .. controls (333,130.13) and (336.13,127) .. (340,127) .. controls (343.87,127) and (347,130.13) .. (347,134) .. controls (347,137.87) and (343.87,141) .. (340,141) .. controls (336.13,141) and (333,137.87) .. (333,134) -- cycle ;
\draw  [fill={rgb, 255:red, 255; green, 255; blue, 255 }  ,fill opacity=1 ] (283,134) .. controls (283,130.13) and (286.13,127) .. (290,127) .. controls (293.87,127) and (297,130.13) .. (297,134) .. controls (297,137.87) and (293.87,141) .. (290,141) .. controls (286.13,141) and (283,137.87) .. (283,134) -- cycle ;
\draw  [dash pattern={on 0.84pt off 2.51pt}] (264,96) .. controls (264,87.16) and (271.16,80) .. (280,80) -- (588,80) .. controls (596.84,80) and (604,87.16) .. (604,96) -- (604,144) .. controls (604,152.84) and (596.84,160) .. (588,160) -- (280,160) .. controls (271.16,160) and (264,152.84) .. (264,144) -- cycle ;
\draw  [fill={rgb, 255:red, 255; green, 255; blue, 255 }  ,fill opacity=1 ] (440,103.5) .. controls (440,99.36) and (443.36,96) .. (447.5,96) .. controls (451.64,96) and (455,99.36) .. (455,103.5) .. controls (455,107.64) and (451.64,111) .. (447.5,111) .. controls (443.36,111) and (440,107.64) .. (440,103.5) -- cycle ;
\draw  [fill={rgb, 255:red, 255; green, 255; blue, 255 }  ,fill opacity=1 ] (470,182) .. controls (470,177.86) and (473.36,174.5) .. (477.5,174.5) .. controls (481.64,174.5) and (485,177.86) .. (485,182) .. controls (485,186.14) and (481.64,189.5) .. (477.5,189.5) .. controls (473.36,189.5) and (470,186.14) .. (470,182) -- cycle ;
\draw  [fill={rgb, 255:red, 255; green, 255; blue, 255 }  ,fill opacity=1 ] (489,103.5) .. controls (489,99.36) and (492.36,96) .. (496.5,96) .. controls (500.64,96) and (504,99.36) .. (504,103.5) .. controls (504,107.64) and (500.64,111) .. (496.5,111) .. controls (492.36,111) and (489,107.64) .. (489,103.5) -- cycle ;
\draw  [fill={rgb, 255:red, 255; green, 255; blue, 255 }  ,fill opacity=1 ] (561,103.5) .. controls (561,99.36) and (564.36,96) .. (568.5,96) .. controls (572.64,96) and (576,99.36) .. (576,103.5) .. controls (576,107.64) and (572.64,111) .. (568.5,111) .. controls (564.36,111) and (561,107.64) .. (561,103.5) -- cycle ;
\draw  [fill={rgb, 255:red, 255; green, 255; blue, 255 }  ,fill opacity=1 ] (541,181.5) .. controls (541,177.08) and (544.58,173.5) .. (549,173.5) .. controls (553.42,173.5) and (557,177.08) .. (557,181.5) .. controls (557,185.92) and (553.42,189.5) .. (549,189.5) .. controls (544.58,189.5) and (541,185.92) .. (541,181.5) -- cycle ;
\draw  [dash pattern={on 0.84pt off 2.51pt}]  (539,105) -- (519,105) ;
\draw  [dash pattern={on 0.84pt off 2.51pt}]  (515,187.5) -- (495,187.5) ;
\draw  [dash pattern={on 0.84pt off 2.51pt}] (382,176.8) .. controls (382,171.39) and (386.39,167) .. (391.8,167) -- (582.2,167) .. controls (587.61,167) and (592,171.39) .. (592,176.8) -- (592,206.2) .. controls (592,211.61) and (587.61,216) .. (582.2,216) -- (391.8,216) .. controls (386.39,216) and (382,211.61) .. (382,206.2) -- cycle ;

\draw (440,85) node [anchor=north west][inner sep=0.75pt]    {$a_{1}$};
\draw (490,85) node [anchor=north west][inner sep=0.75pt]    {$a_{2}$};
\draw (560,85) node [anchor=north west][inner sep=0.75pt]    {$a_{k}$};
\draw (425,192) node [anchor=north west][inner sep=0.75pt]    {$b_{1}$};
\draw (475,192) node [anchor=north west][inner sep=0.75pt]    {$b_{2}$};
\draw (545,192) node [anchor=north west][inner sep=0.75pt]    {$b_{k}$};
\draw (285,115) node [anchor=north west][inner sep=0.75pt]    {$v_{1}$};
\draw (380,115) node [anchor=north west][inner sep=0.75pt]    {$v_{n}{}_{-2k}$};
\draw (335,115) node [anchor=north west][inner sep=0.75pt]    {$v_{2}$};
\draw (225,86) node [anchor=north west][inner sep=0.75pt]   [align=left] {Alice};
\draw (345,170) node [anchor=north west][inner sep=0.75pt]   [align=left] {Bob};

\end{tikzpicture}
    \caption{Reduction from \textsc{Pointer Chasing} to \textsc{Depth First Pointer Chasing}.}
    \label{fig:DFPC}
\end{figure}

The graph is a tree $T$ on $n$ vertices, composed of a path $(v_1,\dots,v_{n-2k})$, and $2k$ leaves $v_{n-2k+1}, \dots, v_n$, all adjacent to $v_{n-2k}$. Node $v_1$ is called the root, and node $v_{n-2k}$ is said central. Note that the ordering obtained by taking $\mbox{DFS}(v_i)=i$ is a depth-first search of~$T$, rooted at $v_1$. 
The root $v_1$ is given value $x_0$ as input. If $i \leq n - 2k$, then function~$f_i$ is merely the identity function~$f$ (i.e., $f(x)=x$ for all~$x$). 
For every $j \in [k]$, let $a_j  = v_{n-k+2j-1}$, and $b_j = v_{n-k+2j}$. All nodes $b_j$ get as input the function~$f_B$, and all nodes $a_j$ get the function~$f_A$. Observe that the output of \textsc{Depth First Pointer Chasing} on this instance is precisely the same as the output of the initial instance of \textsc{Pointer Chasing}. Indeed,  $f_{n-2k} \circ f_{n-2k - 1} \circ \cdots \circ f_1$ is the identity function, and  the sequence $f_n \circ f_{n-1} \circ \cdots \circ f_{n-2k+2} \circ f_{n-2k+1}$ alternates nodes of ``type'' $a_j$ with nodes of ``type'' $b_j$, for decreasing values of $j \in [k]$, and thus corresponds to $f_A \circ f_B \circ \cdots \circ f_A \circ f_B$, where the pair $f_A \circ f_B$ is repeated $k$ times, exactly as in problem \textsc{Pointer Chasing}.

We can now apply Round-Efficient Simulation Lemma. Let 
Alice control all vertices $a_j$, for all $j \in [k]$, and vertices $v_1,\dots,v_{n-2k}$. Let Bob control vertices $b_j$, for all $j \in [k]$. See Fig.~\ref{fig:DFPC}. Note that Alice and Bob can construct the subgraph that they control, based only on their input in the considered \PC{} instance, and they both now value $k$. 

\begin{claim}
If there exists a CONGEST algorithm $\mathcal{A}$ for \DFPC{} on $n$-node graphs performing in $R_{\mathcal{A}}$ rounds with node-activation smaller than~$2k$, then \textsc{Pointer Chasing} can be solved by a two-party protocol $\mathcal{P}$ in less than $2k$ rounds, with communication complexity $O(k^4 \log n \log R_{\mathcal{A}})$ bits.
\end{claim}

The claim directly follows from Lemma~\ref{lem:simu2}. Indeed, by construction, $\partial V_A = 1$ and $\partial V_B = k$. Since we assumed $\NA(\mathcal{A}) < 2k$, the two-way protocol $\mathcal{P}$ provided by Lemma~\ref{lem:simu2} solves the  \PC{} instance in less than $2k$ rounds, and uses $O(k^4 \log n \log R_{\mathcal{A}})$ bits.  

By Lemma~\ref{lem:PointerChasing}, we must have $k^4 (\log n +\log R_{\mathcal{A}}) \in \Omega(n - k \log n)$. Therefore, if our CONGEST algorithm $\mathcal{A}$ has polynomially many rounds, we must have  $k \in \Omega\left(\UpB\right)$. Since our graph has maximum degree $\Delta = 2k+1$, the conclusion follows. 
\qed    
\end{proof}

\section{Conclusion}

In this paper, we have mostly focused on the round complexity of (deterministic) \emph{frugal} algorithms solving general graph problems in the LOCAL or CONGEST model. It might be interesting to consider specific classical problems. As far as ``local problems'' are concerned, i.e., for locally checkable labeling (LCL) problems, we have shown that MIS and $(\Delta+1)$-coloring admit frugal algorithms with polynomial round complexities. It is easy to see, using the same arguments, that problems such as maximal matching share the same properties. It is however not clear that the same holds for $(2\Delta-1)$-edge coloring. 

\begin{openprob}
Is there a (node or edge) frugal algorithm solving $(2\Delta-1)$-edge-coloring with round complexity $O(\mbox{\rm poly}(n))$ in the CONGEST model? 
\end{openprob}

In fact, it would be desirable to design frugal algorithms with sub-polynomial round complexities for LCL problems in general. In particular: 

\begin{openprob}
Is there a (node or edge) frugal algorithm solving {\sc mis} or $(\Delta+1)$-coloring with round complexity $O(\mbox{\rm polylog}(n))$ in the LOCAL model? 
\end{openprob}

The same type of questions can be asked for global problems. In particular, it is known that MST has no ``awake frugal'' algorithms, as MST has awake complexity $\Omega(\log n)$, even in the LOCAL model. In contrast, frugal algorithms for MST do exist as far as node-activation complexity is concerned. The issue is about the round complexities of such algorithms. 

\begin{openprob}
Is there a (node or edge) frugal algorithm solving {\sc mst} with round complexity $O(\mbox{\rm poly}(n))$ in the CONGEST model? 
\end{openprob}

Another intriguing global problem is depth-first search ({\sc dfs}), say starting from an identified node. This can be performed by an edge-frugal algorithm performing in a linear number of rounds in CONGEST. However, it is not clear whether the same can be achieved by a node-frugal algorithm. 

\begin{openprob}
Is there a node-frugal algorithm solving {\sc dfs} with round complexity $O(\mbox{\rm poly}(n))$ in the CONGEST model? 
\end{openprob}

Finally, we have restricted our analysis to \emph{deterministic} algorithms, and it might obviously be worth considering \emph{randomized} frugal algorithms as well.  

\bibliographystyle{splncs04}
\bibliography{biblioActivation}

\paragraph{Acknowledgements.} The authors are thankful to Benjamin Jauregui for helpful discussions about the sleeping model.

\end{document}